\documentclass[journal,10pt]{IEEEtran}
\IEEEoverridecommandlockouts
\usepackage[utf8]{inputenc}
\usepackage[T1]{fontenc} 
\usepackage{graphicx}
\usepackage{listings}
\usepackage{float}
\usepackage{stfloats}
\usepackage{amsmath,amssymb}
\usepackage{verbatim}
\usepackage{algorithm}
\usepackage{algorithmicx}
\usepackage{algpseudocode}
\usepackage{fancyvrb}
\usepackage{verbatim}
\usepackage{mathtools}
\usepackage{bm}
\usepackage{color}
\usepackage{float}
\usepackage{pifont} 
\usepackage{cite}
\usepackage{scalerel}
\usepackage{tikz}
\usepackage[bookmarks=false]{hyperref}
\usetikzlibrary{svg.path}

\newtheorem{proposition}{Proposition}

\title{\huge{Formation Control for CRLB-Optimal Cooperative Sensing \\in Low-Altitude Wireless Networks}}

\author{
    Jun Wu, \IEEEmembership{Student Member, IEEE}, Haijia Jin, \IEEEmembership{Student Member, IEEE}, Nanchi Su, \IEEEmembership{Member, IEEE}, \\  Jinna Li, \IEEEmembership{Senior Member, IEEE}, Haoyuan Pan, \IEEEmembership{Member, IEEE}, and Tse-Tin~Chan, \IEEEmembership{Member, IEEE}
  
\thanks{
  (\textit{Corresponding author: Tse-Tin Chan}) J. Wu is with the School of Automation and Intelligent  Manufacturing, Southern University of Science and Technology, Shenzhen 518055, China, and with the Department of Mathematics and Information Technology, The Education University of Hong Kong, Hong Kong, China.
  H. Jin is with the School of Automation and Intelligent  Manufacturing, Southern University of Science and Technology, Shenzhen 518055, China (e-mail: jinhj2024@mail.sustech.edu.cn).
  N. Su is with Guangdong Provincial Key Laboratory of Aerospace Communication and Networking Technology, Harbin Institute of Technology (Shenzhen), Shenzhen 518055, China, and with the Department of Mathematics and Information Technology, The Education University of Hong Kong, Hong Kong, China (e-mail: sunanchi@hit.edu.cn). 
  J. Li is with the School of Information and Control Engineering, Liaoning Petrochemical University, Fushun, 113001, China (e-mail: jnli@lnpu.edu.cn).
  H. Pan is with the College of Computer Science and Software Engineering, Shenzhen University, Shenzhen, China (e-mail: hypan@szu.edu.cn).
  T.-T.~Chan is with the Department of Mathematics and Information Technology, The Education University of Hong Kong, Hong Kong, China (e-mail: tsetinchan@eduhk.hk).
  }
  \vspace{-0.3in}
  }

\begin{document}
\maketitle
\begin{abstract}
Cooperative sensing with uncrewed aerial vehicles (UAVs) is a key enabler for low-altitude wireless networks (LAWNs), where sensing accuracy critically depends on the spatial configuration of the UAV formation. In this paper, we study formation design and control for Cramér–Rao lower bound (CRLB)-optimal cooperative target sensing. We first establish a sensing performance model based on range measurements and derive the Fisher information matrix (FIM) of the target location. By adopting the A-optimality criterion, we analytically characterize the formation geometry that minimizes the CRLB of the estimation error. The optimal formation is shown to exhibit isotropic Fisher information in the horizontal plane, leading to a regular polygon geometry with an elevation angle determined by the tradeoff between path loss and geometric diversity. Building on this result, we further develop a distributed formation control strategy that steers UAVs from arbitrary initial deployments toward the sensing-optimal configuration while maintaining formation motion and obstacle avoidance. Numerical results demonstrate that the proposed scheme consistently outperforms benchmark formations in terms of CRLB and achieves reliable convergence under practical constraints.
\end{abstract}
\begin{IEEEkeywords}
Low-altitude wireless networks, cooperative sensing, CRLB, formation geometry, distributed formation control.
\end{IEEEkeywords}

\vspace{-0.2in}
\section{Introduction}
 Low-altitude wireless networks (LAWNs) based on uncrewed aerial vehicles (UAVs) have emerged as a flexible platform for sensing, communication, and intelligence acquisition in applications such as surveillance and emergency response. Owing to their mobility and adaptive deployment capability, UAV swarms can cooperatively observe targets from multiple spatial perspectives, thereby significantly enhancing sensing accuracy compared with a single platform \cite{yuan2025ground,wu2025low}.

In cooperative UAV sensing, the relative geometry among sensing nodes plays a fundamental role in determining estimation performance \cite{9097832}. Even with identical sensing hardware and transmit power, different spatial formations can yield substantially different information gains due to geometric dilution and angular diversity. As a result, geometry-aware deployment has become a key issue in UAV-enabled sensing systems, particularly in LAWNs, where path-loss effects and limited onboard resources further amplify the impact of spatial configuration. A principled way to quantify sensing performance is through the Cramér–Rao lower bound (CRLB), which links the minimum achievable mean-squared error (MSE) to the Fisher information matrix (FIM). For example, \cite{8786125} analytically characterized optimal sensor placement for time-of-arrival (ToA)-based localization by deriving closed-form CRLB-optimal geometries. To extend static placement to dynamic scenarios, \cite{10158322} studied a LAWN-enabled system where UAV trajectories were optimized to minimize the CRLB for cooperative target tracking. However, existing studies largely optimize sensor placement or individual trajectories, without explicitly addressing how multiple UAVs should cooperatively form and maintain a sensing-optimal geometric configuration. Moreover, practical UAV sensing missions are subject to additional constraints such as obstacle avoidance and coordinated motion \cite{10557642}. In such environments, sensing performance depends not only on the instantaneous spatial geometry but also on the ability of the UAV swarm to dynamically converge to and sustain a desired formation. These considerations limit the direct applicability of existing placement-based results and motivate the need for formation-level design and distributed control mechanisms that jointly regulate geometry, coordination, and safety.

Motivated by the above, this paper investigates optimal sensing-oriented formation design and distributed control in LAWNs. In contrast to existing approaches, our focus is on jointly characterizing the sensing-optimal formation geometry and realizing it through distributed coordination under physical safety constraints. The main contributions of this paper are summarized as follows:
\begin{itemize}
    \item We analytically characterize the sensing-optimal formation, showing that the CRLB is minimized when the FIM is isotropic, leading to a regular polygon geometry with a common elevation angle determined by the tradeoff between path-loss effects and geometric conditions.
    
    \item Subsequently, a distributed formation control strategy is developed to steer UAVs from arbitrary initial deployments toward the desired formation while ensuring coordinated motion and obstacle avoidance using only local neighbor information.
    
    \item Simulation results are provided to demonstrate that the proposed formation design and control strategy consistently outperform benchmark configurations in terms of CRLB and formation convergence behavior.
\end{itemize}
The rest of this paper is organized as follows. Sec. \ref{sec2} introduces the system model, while Sec. \ref{sec3} derives the optimal UAV formation geometry based on CRLB analysis. Sec. \ref{sec4} presents the distributed formation control strategy. Numerical results are provided in Sec. \ref{sec5}, followed by conclusions in Sec. \ref{sec6}.

 
 \vspace{-0.1in}
 \section{System Model} \label{sec2}

As illustrated in Fig.~\ref{fig:sce}, we consider a typical LAWN in which a formation of $M\ge3$ UAVs cooperatively performs sensing of an unknown ground target. To ensure reliable sensing performance, the UAVs are required to maintain an optimized formation that fully exploits spatial diversity. In this work, our goal is twofold. First, the optimal formation shape is obtained by optimizing the cooperative sensing performance, followed by enforcing the UAVs to align their configurations with this target formation, starting from arbitrary initial deployments.

Without loss of generality, a three-dimensional operating environment is considered, in which all UAVs fly at a fixed altitude $H$. Let $\mathcal{F^*}$ denote the final desired sensing formation, in which the horizontal coordinates of UAV $m$ are given by $\mathbf{q}_{m}^{\mathrm{e}}=[x_m^\mathrm{e}, y_m^\mathrm{e}]^\top$. Furthermore, the ground target to be sensed is located at $\mathbf{s}=[x_s,y_s]^\top\in\mathbb{R}^2$, such that the distance between the $m$-th UAV and the target can be expressed as 
\begin{equation}
\begin{split}d_{m}  =\sqrt{\parallel\mathbf{q}^{\mathrm{e}}_{m}-\mathbf{s}\parallel^2+H^{2}}, \forall m.
\end{split}
\end{equation}
For UAV $m$, the round-trip delay measurement $\tau_m$ is collected and converted into a range through ${d}_m=\tau_m c/2$, where $c$ is the speed of light. The range estimates are affected by measurement noise, such that the measurement of $d_{m}$ can be modeled as
\begin{equation}
\hat{d}_{m}=d_{m}+z_{m}, \forall m,\label{eq:d_measure}
\end{equation}
where $z_{m}\sim\mathcal{N}\left(0,\sigma_{m}^{2}\right)$ is the additive white Gaussian noise (AWGN) with zero mean and variance $\sigma_{m}^{2}$. The distance vector is denoted by $\mathbf{d} = [d_1, \dots, d_M]^\top$.   In general, the variance $\sigma_{m}^{2}$ is inversely proportional to the signal-to-noise ratio (SNR) at the UAV, i.e.,  $\mathrm{SNR}_{m}=\frac{p_{m}G_{\textrm{p}}h_{m}}{\sigma_{0}^{2}}$,
where $p_m$ is the transmit power allocated to UAV $m$, $G_{\textrm{p}}$ is the processing gain, and $\sigma_0^2$ is the receiver noise power. For a two-way sensing channel, the power gain is modeled as a quartic distance loss, i.e., $h_{m}\triangleq\frac{\beta_{0}}{\left(d_{m}\right)^{4}}$ with $\beta_{0}\triangleq\frac{G_{\textrm{T}}G_{\textrm{r}}\sigma_{\textrm{rcs}}\lambda^{2}}{\left(4\pi\right)^{3}}$ representing the channel power at a unit reference distance. Here, $G_{\textrm{T}}$ and $G_{\textrm{r}}$ denote the transmit and receive antenna gains, respectively, $\sigma_{\textrm{rcs}}$ is the radar cross section (RCS), and $\lambda$ is the carrier wavelength. Consequently, $\sigma_{m}^{2}$ can be particularized as $\sigma_{m}^{2}=\frac{\kappa\sigma_{0}^{2}\left(d_{m}\right)^{4}}{p_{m}G_{\textrm{p}}\beta_{0}}$, with $\kappa$ being a constant system parameter \cite{11045436}. 

\begin{figure}
    \centering
    \includegraphics[width=\linewidth]{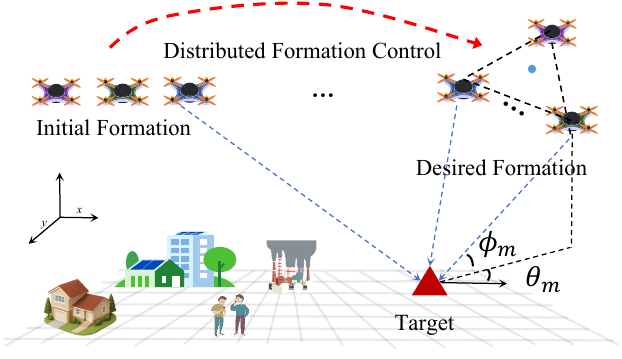}
    \caption{The considered LAWN system where a formation is dispatched to  sense the target cooperatively.}
    \label{fig:sce}
    \vspace{-0.5cm}
\end{figure}
During the transition of the formation toward $\mathcal{F^*}$, the planar motions are controlled.
To facilitate coordination, one UAV is selected as the formation leader, say UAV $1$, while the remaining UAVs are followers.
The set of UAVs is indexed by $\mathcal{M}=\{1,2,\ldots,M\}$.  At discrete time instant $k$, the planar position of UAV $m\in\mathcal{M}$ is denoted by $\mathbf{q}_{m,k}\in\mathbb{R}^{2}$. The motion of the UAV formation is modeled as the superposition of a global translational movement of the formation center and individual local adjustments that shape the formation geometry. Accordingly, the position dynamics of UAV $m$ are described by the following first-order Gaussian Markov model\cite{10557642}, i.e.
\begin{equation}
\mathbf{q}_{m,k+1}
=\mathbf{q}_{m,k}
+\mathbf{u}_{m,k}
+\mathbf{g}_{m,k}\Delta T
+\mathbf{n}_{m,k},
\label{eq:state}
\end{equation}
where $\mathbf{g}_{m,k}\in\mathbb{R}^2$ represents the velocity associated with the global motion at time $k$, $\mathbf{u}_{m,k}\in\mathbb{R}^2$ denotes the formation control input of UAV $m$, $\Delta T$ is the sampling interval between any two adjacent instants, and $\mathbf{n}_{m,k}\sim\mathcal{N}(\mathbf{0},\sigma^2\mathbf{I})$ accounts for random motion disturbance, where $\sigma^2$ denotes the disturbance variance. 


\vspace{-0.1in}
\section{ CRLB-Optimal Formation Geometry}  \label{sec3}
In this section, we derive the formation geometry $\mathcal{F^*}$ for optimally sensing the ground target $\mathbf{s}$.   In general, MSE is employed as the performance metric for unbiased estimation. To achieve MSE optimization, the A-optimality criterion, which minimizes the average estimation variance, is adopted. For ease of exposition, we assume that a prior coarse estimate of the target's location $\mathbf{s}$ is available.  

To evaluate the sensing performance, we first construct the FIM of the range measurements $\mathbf{J}_{\mathbf{d}} $, such that the FIM of coordinates can be achieved via the chain rule, i.e., 
$  \mathbf{J}_\mathbf{s} = \mathbf{Q}^\top \mathbf{J}_{\mathbf{d}} \mathbf{Q},$
where  $\mathbf{Q} \in \mathbb{R}^{M \times 2}$ is the Jacobian matrix given by 
\begin{equation}
\mathbf{Q} = \frac{\partial \mathbf{d}}{\partial \mathbf{s}} = \begin{bmatrix} \frac{\partial d_1}{\partial x_s} & \frac{\partial d_1}{\partial y_s} \\ \vdots & \vdots \\ \frac{\partial d_M}{\partial x_s} & \frac{\partial d_M}{\partial y_s} \end{bmatrix}.
\end{equation}
Based on \eqref{eq:d_measure}, the likelihood function for the range measurement of the $m$-th UAV is given by $p(\hat{d}_m | d_m) = \frac{1}{\sqrt{2\pi \sigma_m^2}} \exp\left(-\frac{(\hat{d}_m - d_m)^2}{2\sigma_m^2}\right)$. Assuming the measurement noises among different UAVs are independent, the FIM $\mathbf{J}_{\mathbf{d}} \in \mathbb{R}^{M \times M}$ is therefore a diagonal matrix. Each diagonal element $[\mathbf{J}_{\mathbf{d}}]_{m,m}$ accounts for information from both the measurement mean and the variance, i.e., 
\begin{equation}
\begin{split}
[\mathbf{J}_{\mathbf{d}}]_{m,m} &= -\mathbb{E}\left[ \frac{\partial^2 \ln p(\hat{d}_m | d_m)}{\partial d_m^2} \right] \\
&= \frac{1}{\sigma_m^2} \left( \frac{\partial d_m}{\partial d_m} \right)^2 + \frac{1}{2\sigma_m^4} \left( \frac{\partial \sigma_m^2}{\partial d_m} \right)^2.
\end{split}
\end{equation}
With some algebraic manipulations, we arrive at 
\begin{equation}
[\mathbf{J}_{\mathbf{d}}]_{m,m} = \frac{1}{\sigma_m^2} + \frac{1}{2\sigma_m^4} \left( \frac{4\sigma_m^2}{d_m} \right)^2 = \frac{p_{m}G_{\textrm{p}}\beta_{0}}{\kappa\sigma_{0}^{2}d_{m}^{4}} + \frac{8}{d_{m}^{2}}.
\end{equation}
To facilitate the formation shape derivation, we have the geometric relationship between UAV $m$ and the target as follows
\begin{equation}
\begin{split}
    \frac{\partial d_m}{\partial x_s} &= \frac{x_s-x_m^\mathrm{e} }{d_m} = \cos \phi_m \cos \theta_m, \\
    \frac{\partial d_m}{\partial y_s} &= \frac{ y_s-y_m^\mathrm{e} }{d_m} = \cos \phi_m \sin \theta_m,
\end{split}
\end{equation}
where $\theta_m\in[0,2\pi]$ is the azimuth angle and $\phi_m\in (0,\pi/2)$ is the elevation angle of UAV $m$ relative to the target. 
 Then, the FIM for the target coordinates $\mathbf{s}$ is calculated as:
\begin{equation}
\mathbf{J}_\mathbf{s} = \mathbf{Q}^\top \mathbf{J}_{\mathbf{d}} \mathbf{Q} = \begin{bmatrix} J_{xx} & J_{xy} \\ J_{xy} & J_{yy} \end{bmatrix}.
\end{equation}
The individual components of $\mathbf{J}_\mathbf{s} $ are written by 
\begin{equation}\label{fimj}
\begin{aligned}
J_{xx} &= \sum_{m=1}^{M} w_m  \cos^2 \theta_m, \quad
J_{yy} = \sum_{m=1}^{M} w_m \sin^2 \theta_m, \\
J_{xy} &= \sum_{m=1}^{M} w_m \sin \theta_m \cos \theta_m,
\end{aligned}
\end{equation}
with $w_m=(\frac{p_{m}G_{\textrm{p}}\beta_{0}\sin^4\phi_m}{\kappa\sigma_{0}^{2}H^{4}} + \frac{8\sin^2\phi_m}{H^{2}})\cos^2 \phi_m  $ due to $\sin \phi_m=H/d_m$.
The A-optimality criterion minimizes $\text{tr}(\mathbf{J}_{\mathbf{s}}^{-1})$, i.e., $\arg \min_{\theta_m,\phi_m} \ \text{tr}(\mathbf{J}_{\mathbf{s}}^{-1})$. To this end, we have the following results:

\begin{proposition}
\label{lemma1}
For any fixed $w_m$, the trace of the inverse FIM, $\text{tr}(\mathbf{J}_{\mathbf{s}}^{-1})$, reaches its global minimum if and only if $\mathbf{J}_{\mathbf{s}}$ is a scalar matrix.
\end{proposition} 

\begin{IEEEproof}
Let $\lambda_1, \lambda_2$ be the eigenvalues of the symmetric FIM $\mathbf{J}_{\mathbf{s}}$.
The trace of the inverse FIM is given by
\begin{equation}
    \text{tr}(\mathbf{J}_{\mathbf{s}}^{-1}) = \frac{1}{\lambda_1} + \frac{1}{\lambda_2}.
\end{equation}
According to the arithmetic-harmonic mean inequality, we have
\begin{align}
    \frac{1}{\lambda_1} + \frac{1}{\lambda_2} \geq \frac{4}{\lambda_1 + \lambda_2} = \frac{4}{\text{tr}(\mathbf{J}_{\mathbf{s}})}=\frac{4}{\sum_{m=1}^{M} w_m }, \label{theolowbound}
\end{align}
where equality holds if and only if $\lambda_1 = \lambda_2$.
Since $\mathbf{J}_{\mathbf{s}}$ is symmetric, it is orthogonally diagonalizable as $\mathbf{J}_{\mathbf{s}} = \mathbf{U}\mathbf{\Lambda}\mathbf{U}^T$ with an orthogonal matrix $\mathbf{U}$ and a diagonal matrix $\mathbf{\Lambda}$.
When $\lambda_1 = \lambda_2 = \lambda$, the diagonal matrix becomes $\mathbf{\Lambda} = \lambda \mathbf{I}$, which yields
$\mathbf{J}_{\mathbf{s}} = \mathbf{U}(\lambda \mathbf{I})\mathbf{U}^T  = \lambda \mathbf{I}$, completing the proof.
\end{IEEEproof}
Proposition \ref{lemma1} suggests that $\text{tr}(\mathbf{J}_{\mathbf{s}}^{-1})$ achieves minimal value $\frac{4}{\sum_{m=1}^{M} w_m }$ if $J_{xx}=J_{yy}$ and $J_{xy}=0$.
In other words, the determination of the optimal formation $\mathcal{F^*}$ can be reformulated as the following optimization problem 
\begin{subequations}
    \label{eq:opt_model}
    \begin{align}
        &\min_{\phi_m,\theta_m} \quad \frac{4}{\sum_{m=1}^{M} w_m} \label{obbj}\\
        &\mathrm{s.t.} \quad \sum_{m=1}^M w_m \cos 2\theta_m = 0,  \sum_{m=1}^M w_m \sin 2\theta_m = 0. \label{constrain1}
    \end{align}
\end{subequations}
In general, problem \eqref{eq:opt_model} is challenging to solve directly due to the potential coupling in \eqref{constrain1}. However, we observe that the objective function \eqref{obbj} depends solely on the elevation angles ${\phi_m}$. \footnote{Note that $w_m$ also depends on the transmit power $p_m$. As this work focuses on sensing-optimal formation geometry, $p_m$ is treated as a constant across all UAVs. Related optimizations can be found in \cite{ding2023multi,10648717,11045436}.} This allows us to determine the optimal solution via a two-step approach, i.e., solving $\phi_m$ that minimizes the theoretical lower bound first, followed by finding $\theta_m$  satisfying \eqref{constrain1}.
Since each UAV can independently choose its own elevation angle,  the variables $\phi_1, \phi_2, \dots, \phi_M$ are decoupled. Then, maximizing $\sum_{m=1}^{M} w_m$ is equivalent to each individual term $w_m$ reaching its respective maximum value, i.e., $\phi_1^*=\phi_2^*=\cdots=\phi_M^*$. To determine the optimal $\phi_m^*$, we set the first-order derivative of $w_m$ with respect to $\phi_m$ equal to zero, yielding 
$\frac{d w_m}{d \phi_m} = \sin 2\phi_m D_m = 0$, where  $D_m = \frac{p_m G_{\textrm{p}} \beta_0}{\kappa \sigma_0^2 H^4} \left( 2 \sin^2 \phi_m \cos^2 \phi_m - \sin^4 \phi_m \right) + \frac{8}{H^2} (\cos^2 \phi_m - \sin^2 \phi_m)$. Given that $\phi_m \in (0, \pi/2)$, it follows that $\sin 2\phi_m \neq 0$, which implies $D_m = 0$. This equation can be further transformed into a quadratic form in terms of $\tan^2 \phi_m$, which can then be solved analytically. Due to space constraints, the detailed derivations are omitted here. Moreover, it is interesting to observe that when $H$ is relatively small, the $H^{-4}$ term dominates, leading $\phi_m^*$ to asymptotically approach $\arctan(\sqrt{2}) \approx 54.74^\circ$. At this stage, the system prioritizes maximizing the projected SNR. As $H$ increases, the $H^{-4}$ term decays rapidly, and the equation becomes governed by the geometric term $(\cos^2 \phi_m - \sin^2 \phi_m) \to 0$. Consequently, $\phi_m^*$ converges to $45^\circ$, shifting the priority from energy preserving to geometric robustness, i.e., $\phi_m^*\in(45^\circ,54.74^\circ)$.  

Subsequently, we derive the optimal azimuth angles ${\theta_m^*}$. Given that the optimal elevation angles are uniform, and thus $w_m = w^*, \forall m $, the weights in \eqref{constrain1} become a common scaling factor. Consequently, the constraints for achieving the minimum CRLB reduce to the symmetry conditions $ \sum_{m=1}^M \cos 2\theta_m = 0$ and $ \sum_{m=1}^M \sin 2\theta_m = 0$, which can be compactly represented in complex form as $\sum_{m=1}^M e^{j 2\theta_m} = 0$. This condition implies that the Fisher information must be isotropically distributed in the horizontal plane to eliminate off-diagonal correlations in the FIM. A standard solution satisfying these equations is the regular $M$-polygon formation, where the azimuth angles are distributed as $    \theta_m^* = \theta_0 + \frac{2\pi(m-1)}{M}, \quad m = 1, \dots, M$, 
where $\theta_0 \in [0, 2\pi)$ is an arbitrary initial rotation angle. In summary, the optimal sensing formation $\mathcal{F^*}$ can be geometrically characterized as an \textit{inverted cone} with its vertex at the target's location. Specifically, all UAVs are distributed uniformly along a horizontal circular orbit at altitude $H$, centered at the target's vertical projection. 
\vspace{-0.1in}
\section{Distributed Formation Control}
 \label{sec4}
In this section, we elaborate on the distributed formation control strategy, starting from arbitrary initial formations. With $\mathcal{F^*}$ at hand, the objective of formation control is to make the relative positions between UAVs satisfy the preset displacement vector set $\mathbb{F}$, which is defined as
\begin{align}
    \mathbb{F} = \left\{ \mathbf{q}_{m,k} - \mathbf{q}_{p,k} = \boldsymbol{\delta}_{pm}, \frac{\mathbf{q}_{m,k+1} - \mathbf{q}_{m,k}}{\Delta T} = \mathbf{v}^*, \forall m, p \in \mathcal{M} \right\},
\end{align}
where $\boldsymbol{\delta}_{pm} \in \mathbb{R}^2$ represents the desired constant displacement of UAV $m$ with respect to UAV $p$, and $\mathbf{v}^* \in \mathbb{R}^2$ denotes the targeted global moving velocity of the entire formation. 

To achieve this in a distributed manner, we first define an undirected communication graph $\mathcal{G}_f = (\mathcal{M}, \mathcal{E}, \mathcal{A})$ for the formation, where $\mathcal{M} = \{1, 2, \dots, M\}$ is the set of UAV nodes and $\mathcal{E} \subseteq \mathcal{M} \times \mathcal{M}$ denotes the edge set, where an edge $(m,p) \in \mathcal{E}$ indicates that UAV $m$ and $p$ can directly exchange information. $\mathcal{A}  \in \mathbb{R}^{M \times M}$ is the binary adjacency matrix with entries $a_{mp}=1$ if $(m,p) \in \mathcal{E}$ and $a_{mp}=0$ otherwise. According to \cite{oh2015survey}, in case that $\mathcal{G}_f$ is connected, the desired formation $\mathcal{F^*}$ can be uniquely characterized.  To guarantee a globally uniform velocity $\mathbf{v}^*$ across all UAVs, we design a consensus mechanism that enables distributed velocity alignment. Assuming that only the leader has prior knowledge of $\mathbf{v}^*$, the velocity update law for the follower UAV $m$ is given by 
\begin{equation}
    \mathbf{v}_{m,k+1} = - \sum_{p \in \mathcal{N}_m\backslash \{1\}} a_{mp} \left( \mathbf{v}_{m,k} - \mathbf{v}_{p,k} \right)+a_{m1} \left( \mathbf{v}_{m,k} - \mathbf{v}^* \right),
\end{equation}
where $\mathbf{v}_{m,k}$ is the desired velocity estimate of UAV $m$ and $\mathcal{N}_m = \{p \in \mathcal{M} | a_{mp}=1\}$ is the neighbor set of UAV $m$. 
If $\mathcal{G}_f$ is connected, all UAVs’ velocities will asymptotically converge to $\mathbf{v}^*$ as $k \to \infty$ \cite{ren2008distributed}, such that the global update can be simply written as $   \mathbf{g}_{m,k}=\mathbf{v}_{m,k}.$

Subsequently, we devise the formation control input $\mathbf{u}_{m,k}$ by first constructing a decentralized cost function for UAV $m$
\begin{equation} \label{costfun}
\begin{split}
        \mathcal{L}_m\left(\mathbb{F}, \mathbf{q}_{m,k}\right)= \sum_{p \in \mathcal{M}} a_{mp} \left\| \mathbf{q}_{m,k} - \mathbf{q}_{p,k} - \boldsymbol{\delta}_{pm} \right\|^2 \\+ \left\| \frac{\mathbf{q}_{m,k+1} - \mathbf{q}_{m,k}}{\Delta T} - \mathbf{v}^* \right\|^2,
    \end{split}
\end{equation}
where the first term penalizes deviations from the desired relative displacements and the second term enforces alignment with the consensus-based velocity. Hence, the total cost function is given by $\mathcal{L}_{\mathrm{tot}}=\sum_{m \in \mathcal{M}} \mathcal{L}_m\left(\mathbb{F}, \mathbf{q}_{m,k}\right)$. 
Motivated by \eqref{costfun}, we design the formation-shape control input $\mathbf{u}_{m,k}$ as a distributed negative-gradient step of the displacement-error term in \eqref{costfun} with respect to $\mathbf{q}_{m,k}$, which yields the control policy
\begin{align} \label{controlinput}
    \mathbf{u}_{m,k} = -\varepsilon\sum_{p \in \mathcal{M}} a_{mp} \left( \mathbf{q}_{m,k} - \mathbf{q}_{p,k} - \boldsymbol{\delta}_{pm} \right),
\end{align}
with $\varepsilon > 0$ being the convergence speed tuning parameter. 
This displacement-based approach leverages linear closed-loop dynamics to guarantee global convergence, where only local neighbor information is required, resulting in low computational complexity and communication overhead.
To prevent obstacle collisions in practice, a repulsive control component $\mathbf{u}^{\mathrm{r}}_{m,k}$ is further integrated into the control architecture \cite{garcia2015repulsive}.  In particular, let  $l_{m,k}$ be the shortest distance from UAV $m$ to the nearest point $\mathbf{x}_{\text{near},k}$ on the obstacle surface in time instant $k$. Then, one can define a  potential function given by 
\begin{align}
    U^{\mathrm{r}}_{m,k} = \frac{1}{2} E^{\text{r}} \left( \frac{1}{l_{m,k}} - \frac{1}{l_{\text{safe}}} \right)^2, \quad \text{if } l_{m,k}< l_{\text{safe}},
\end{align}
where $ E^{\text{r}}$ is the repulsion gain and $l_{\text{safe}}$ is the preset safety radius. The resulting repulsion vector $\mathbf{u}^{\mathrm{r}}_{m,k}$ is calculated as the negative gradient of the potential function given by
\begin{align}
    \mathbf{u}^{\mathrm{r}}_{m,k} = -\nabla U^{\mathrm{r}}_{m,k} = E^{\text{r}} \left( \frac{1}{l_{m,k}} - \frac{1}{l_{\text{safe}}} \right) \frac{1}{l_{m,k}^2} \frac{\vec{\mathbf{l}}_{m,k}}{l_{m,k}},
\end{align}
where $\vec{\mathbf{l}}_{m,k} = \mathbf{q}_{m,k} - \mathbf{x}_{\text{near},k}$ represents the relative displacement vector pointing from the nearest obstacle point to the UAV. For example, regarding rectangular obstacles with boundaries $[x_{\min}, x_{\max}]$ and $[y_{\min}, y_{\max}]$, the nearest point $\mathbf{x}_{\text{near},k} = [x_k, y_k]^T$ is determined by $x_k = \max(x_{\min}, \min({q}^x_{m,k}, x_{\max}))$ and $y_k = \max(y_{\min}, \min({q}^y_{m,k}, y_{\max}))$, with $\mathbf{q}_{m,k}=[{q}^x_{m,k},{q}^y_{m,k}]^\top$.
To accommodate formation navigation through narrow constraints, we further introduce a global scaling factor $\eta_k \in (0, 1]$ that dynamically resizes the target displacement set $\mathbb{F}$ without altering its underlying geometric topology. By incorporating both the scaling factor and the repulsive component, the total synthesized control input \eqref{controlinput} is re-expressed as
\begin{align}
    \mathbf{u}_{m,k} = -\varepsilon \sum_{p \in \mathcal{M}} a_{mp} \left( \mathbf{q}_{m,k} - \mathbf{q}_{p,k} - \eta_k \boldsymbol{\delta}_{pm} \right) + 
    \mathbf{u}^{\mathrm{r}}_{m,k}.
\end{align}
where $\eta_k$ is adaptively reduced when the formation center approaches spatial bottlenecks. Then, the stable state evolution model can be expressed as $  \mathbf{q}_{m,k+1}=\mathbf{q}_{m,k} + \mathbf{u}_{m,k}+\mathbf{v}^*\Delta T+\mathbf{n}_{m,k}$.  By doing so, the control law ensures that the UAV ensemble not only converges to the desired formation from arbitrary initializations but also maintains physical safety.

\vspace{-0.1in}
\section{Numerical Results}  \label{sec5}
In this section, we present extensive simulation results to verify the effectiveness of our proposed algorithm. We initialize $M=6$ UAVs randomly near the origin within a rectangular area. Unless otherwise specified, the primary system parameters are summarized as follows: the flight altitude is set to $H=20\ \text{m}$, transmit power $p_m = 0.1\ \text{W}, \forall m$, and the receiver noise power $\sigma_0^2 = -90 \ \text{dBm}$. The sampling interval is set to $\Delta T=0.1$ s. The ground target is located at $\mathbf{s} = [80, 90]^\top$ m. To ensure flight safety, the obstacle avoidance radius is set to $l_\text{safe} = 5\ \text{m}$ and the tuning parameter is set to $\varepsilon=0.01$. 

Fig.~\ref{fig:crlb_altitude} depicts the CRLB as a function of the UAV flight altitude under different formation configurations, thereby validating the effectiveness of the proposed cooperative sensing–oriented formation design. It is observed that the proposed CRLB-optimal formation $\mathcal{F^*}$ consistently outperforms the benchmark schemes across the entire altitude range. Moreover, the performance advantage becomes more pronounced as the altitude increases, which can be attributed to the fact that higher altitudes amplify path-loss effects and reduce angular diversity, making the sensing accuracy increasingly sensitive to the geometric configuration of the formation. 

\begin{figure}[t]
    \centering
    \includegraphics[width=0.9\linewidth]{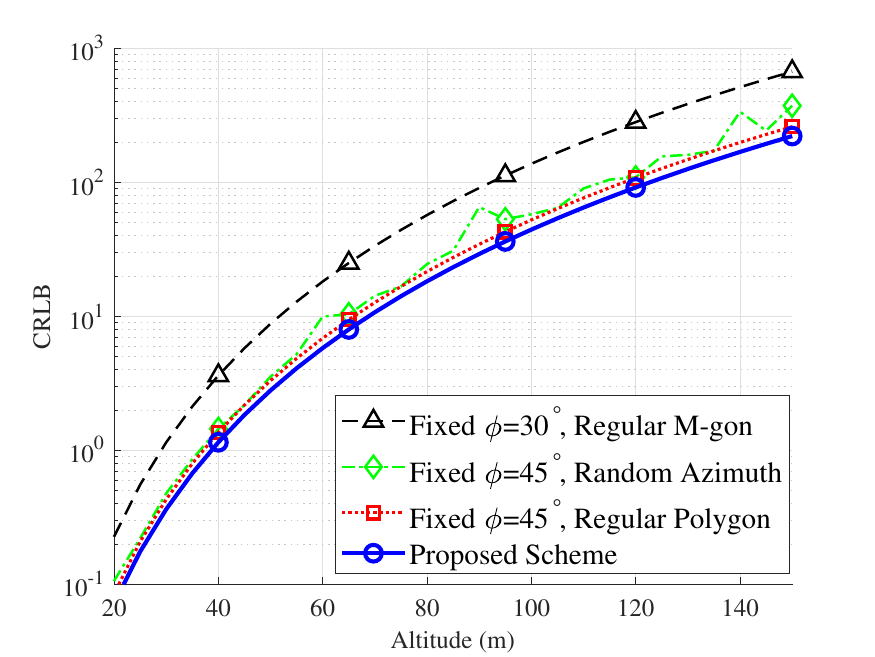}
    \caption{The CRLB vs. flight altitudes under various formation configurations. }
    \label{fig:crlb_altitude}
    \vspace{-0.5cm}
\end{figure}

\begin{figure}
    \centering
    \includegraphics[width=0.9\linewidth]{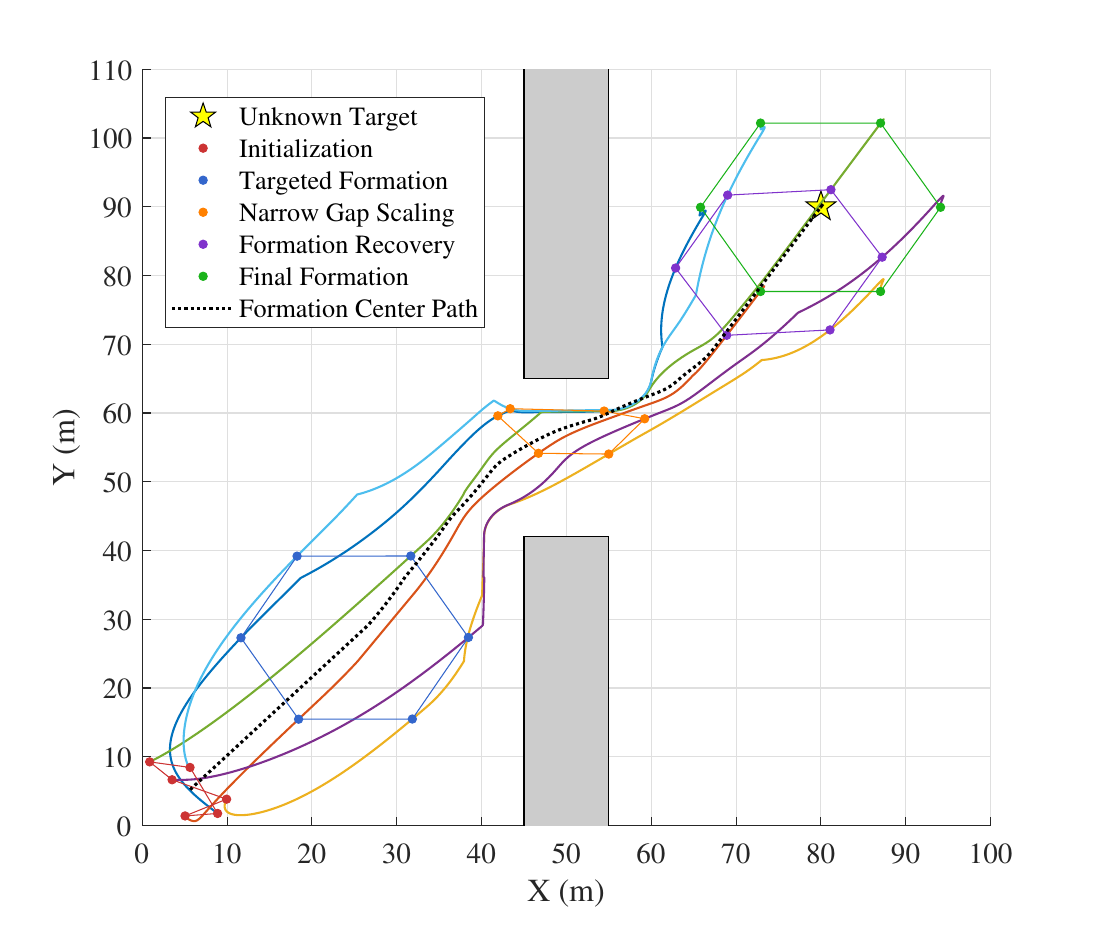}
    \caption{The formation evolution behavior and UAV trajectory in the presence of obstacles. }
    \label{fig:formation_evolution}
    \vspace{-0.6cm}
\end{figure}

In Fig.~\ref{fig:formation_evolution}, we illustrate the evolution of the UAV formation in the presence of obstacles. Beginning from an arbitrary initial deployment, the UAVs first coordinate their motions to establish the desired formation structure. As the formation encounters narrow passages, the UAVs adaptively adjust their relative separations, leading to a temporary deformation and contraction of the formation to satisfy collision-avoidance constraints. Once the obstacles are safely bypassed, the formation progressively expands and reconfigures to recover the prescribed sensing-optimal geometry $\mathcal{F^*}$ and eventually hovers over the target. This evolution process demonstrates that the proposed distributed control strategy effectively couples formation reshaping with cooperative motion.

\begin{figure}[t]
    \centering
    \includegraphics[width=\linewidth]{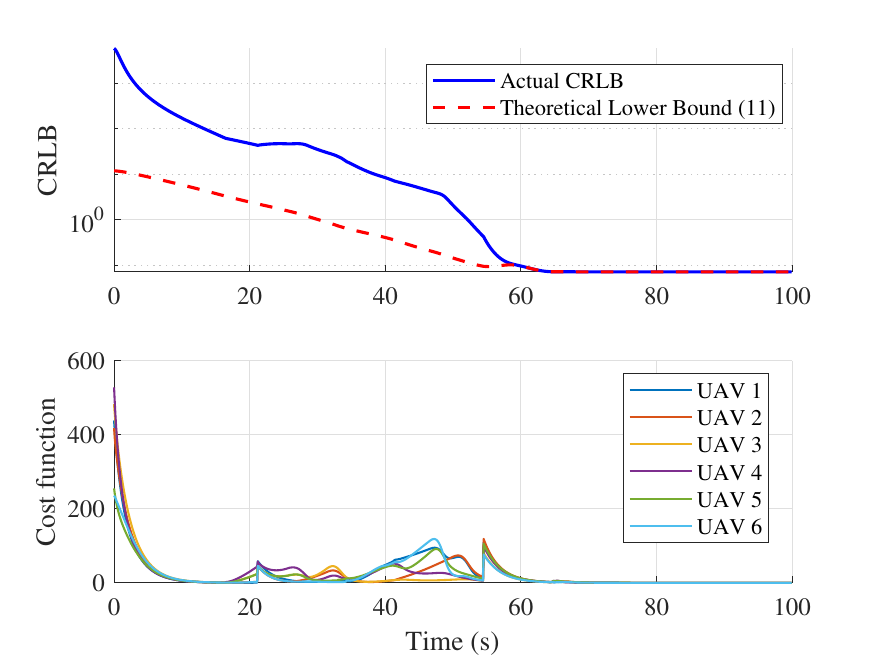}
    \caption{The evaluation of CRLB and cost function \eqref{costfun}. }
    \label{fig:time_varying}
    \vspace{-0.5cm}
\end{figure}

Fig.~\ref{fig:time_varying} depicts the time evolution of the CRLB and the control cost \eqref{costfun}. We observe that the CRLB decreases rapidly as the UAVs adjust their relative positions and the formation approaches the sensing-optimal geometry. As the formation converges, the CRLB stabilizes near the theoretical lower bound \eqref{theolowbound}. In parallel, the control cost generally exhibits a decreasing trend as the formation stabilizes. However, a temporary increase is observed during the interval of approximately 20--60~s, corresponding to the case in which the formation encounters obstacles and undergoes deformation to satisfy collision-avoidance constraints. This deviation from the desired formation increases the control effort.  It is noted that the formation gradually recovers the target configuration, and the control cost subsequently decreases and converges to a steady value. 


\vspace{-0.1in}
\section{Conclusion}  \label{sec6}
This paper studied CRLB-optimal formation design and distributed control for cooperative sensing in LAWNs. By analyzing range-based measurements, the Fisher information matrix for the target location was derived, and the sensing-optimal formation was analytically characterized under the A-optimality criterion, yielding a regular polygon geometry with a common elevation angle. Then, a distributed formation control strategy was developed to navigate UAVs from arbitrary initial positions toward the sensing-optimal configuration while ensuring obstacle avoidance. Extensive simulation results verified the effectiveness of the proposed design.


 \vspace{-0.2in}
\bibliographystyle{IEEEtran}
\bibliography{citation}

\end{document}